\pdfoutput=1
\documentclass[letterpaper, 10 pt, conference]{ieeeconf}  % Comment this line out if you need a4paper

\IEEEoverridecommandlockouts                              % This command is only needed if
\usepackage{graphicx}
\usepackage{dsfont}
\usepackage{amsmath,bm,amssymb}
\usepackage{graphics} % for pdf, bitmapped graphics files
\usepackage{comment}
\usepackage{epsfig} % for postscript graphics files
\usepackage{array}
\usepackage{subcaption}
\usepackage{comment}
\usepackage{multicol}
\usepackage{multirow}
\usepackage{color}
\usepackage{subcaption}
\usepackage{bbm}
\usepackage{mathrsfs}
\usepackage{xparse}
\usepackage{algorithm}
\usepackage{algpseudocode}
\newtheorem{theorem}{Theorem}
\newtheorem{lemma}{Lemma}

\usepackage{mathtools}

\usepackage[colorinlistoftodos]{todonotes}
\usepackage{etoolbox}
\newtoggle{comments}
\toggletrue{comments}

\makeatletter
\newcommand{\algmargin}{\the\ALG@thistlm}
\makeatother
\newlength{\whilewidth}
\algdef{SE}[parWHILE]{parWhile}{EndparWhile}[1]
{\parbox[t]{\dimexpr\linewidth-\algmargin}{%
		\hangindent\whilewidth\strut\algorithmicwhile\ #1\ \algorithmicdo\strut}}{\algorithmicend\ \algorithmicwhile}%
\algnewcommand{\parState}[1]{\State%
	\parbox[t]{\dimexpr\linewidth-\algmargin}{\strut #1\strut}}

\iftoggle{comments}{
	\newcommand {\pli}[1]{{\color{blue}{~Pan: #1}\normalfont}}
	\newcommand {\alberto}[1]{{\color{orange}{~Alberto: #1}\normalfont}}
	\newcommand {\bjin}[1]{{\color{violet}{~Baihong: #1}\normalfont}}
	\newcommand {\rxiong}[1]{{\color{violet}{~Ruoxuan: #1}\normalfont}}
}{
	\newcommand {\pli}[1]{{}}
	\newcommand {\bjin}[1]{{}}
	\newcommand {\alberto}[1]{{}}
	\newcommand {\rxiong}[1]{{}}
}

\title{\LARGE \bf
A tractable ellipsoidal approximation for voltage regulation problems
}

\author{Pan Li$^{1}$, Baihong Jin$^{2}$, Ruoxuan Xiong$^{3}$, Dai Wang$^{4}$, Alberto Sangiovanni-Vincentelli$^{2}$ and Baosen Zhang$^{5}$% <-this % stops a space
\thanks{\textcolor{black}{This work is supported in part  by NSF grant 1544160, and in part by the National Research Foundation of Singapore through a grant to the Berkeley Education Alliance for Research in Singapore (BEARS) for the Singapore-Berkeley Building Efficiency and Sustainability in the Tropics (SinBerBEST) program.}}% <-this % stops a space
\thanks{$^{1}$Pan Li is with Facebook Inc., 1 Hacker Way, Menlo Park, CA 94025. This work was done when she was affiliated with University of Washington, Seatte, WA 98195
        {\tt\small lipan.uw@gmail.com}}%
\thanks{$^{2}$Baihong Jin and Alberto Sangiovanni-Vincentelli are with the Department of Electrical Engineering and Computer Science, University of California, Berkeley,
        CA 94720
        {\tt\small \{bjin,alberto\}@berkeley.edu}}%
\thanks{$^{3}$Ruoxuan Xiong is with the Department of Management Science and Engineering, Stanford University, Stanford, CA 94305
    	{\tt\small rxiong@stanford.edu}}
\thanks{$^{4}$Dai Wang is with Tesla Inc., 3500 Deer Creek Rd., Palo Alto, CA 94304
	{\tt\small daiwang@tesla.com}}
\thanks{$^{5}$ Baosen Zhang are with the Department of Electrical Engineering,
	University of Washington, Seatte, WA 98195
	{\tt\small  zhangbao@uw.edu}}
}

\begin{document}

\maketitle
\thispagestyle{empty}
\pagestyle{empty}

%%%%%%%%%%%%%%%%%%%%%%%%%%%%%%%%%%%%%%%%%%%%%%%%%%%%%%%%%%%%%%%%%%%%%%%%%%%%%%%%
\begin{abstract}

We present a machine learning approach to the solution of chance constrained optimizations in the context of voltage regulation problems in power system operation. The novelty of our approach resides in approximating the feasible region of uncertainty with an ellipsoid. We formulate this problem using a learning model similar to Support Vector Machines (SVM) and propose a sampling algorithm that efficiently trains the model. We demonstrate our approach on a voltage regulation problem using standard IEEE distribution test feeders.

\end{abstract}

%%%%%%%%%%%%%%%%%%%%%%%%%%%%%%%%%%%%%%%%%%%%%%%%%%%%%%%%%%%%%%%%%%%%%%%%%%%%%%%%
\section{INTRODUCTION}\label{sec:intro}

%\subsection{Voltage control in power systems}
Voltage regulation is crucial to maintain voltage stability in distribution systems under different operating conditions~\cite{LietAl2014}. Traditionally, voltage is regulated by changing reactive power through tap-changing transformers and switched capacitors~\cite{ZhangetAl2013}. With the advances of distributed energy resources (DERs, i.e., electric vehicles~\cite{WangEtAl2016}, PV panels~\cite{KanchevEtAl2011,Zhang2015})
reactive power support to regulate voltage may be available. These control strategies have been analyzed in literature by either centralized schema~\cite{FarivarEtAl2012pes,ValvEtAl2013} or distributed algorithms~\cite{ZhuEtAl2016,LietAl2014, ZhangetAl2013,SulcEtAl2016}.

DERs bring significant uncertainty and fast variations in real power to the distribution system~\cite{TuritsynetAl2011,Zhang2015,Robbins13,jin2015contract}. Since most distribution systems do not yet have real-time communication capabilities, a decision made must be valid for a set of possible conditions. For example, in a system with solar PVs, the buses communicate with the feeder (or some other coordinators) every 5 minutes to receive a command for setting its reactive power, but the changes in solar irradiation result in sub-minute timescales changes in its active power. In this paper, we consider a centralized control framework, where periodic control signals are designed to regulate voltages for the entirety of the period in the presence of randomness.

A natural framework to handle the uncertainties introduced by the fast variation in the real power output of the DERs is chance constraints. Indeed, chance constraints bound the probability of voltage constraint violations~\cite{HajianEtAl2012}. Chance constraints may yield difficult optimization problems, since algorithms may have to be designed on a case-by-case basis. Authors in~\cite{YuryEtAl2016, ZhangEtAl2011,WuEtAl2014} offer different relaxation techniques to break chance constraints into simpler ones and formulate stochastic optimizations in power systems as conic problems. However, this relaxation can be very conservative and may not apply to large-scale problems~\cite{BakerEtAl2017}. In~\cite{MartinezEtAl2014}, the authors probe the boundary of the feasible region by the so-called ``$p$-efficient points''. However, the proposed procedure requires solving a mixed integer programming at each iteration. In~\cite{LiEtAl2017}, the authors adopt a gradient-like algorithm that efficiently finds a sub-optimal solution. %To sum up, most of the literature aim to solve chance constrained problems by either heuristics or a tractable alternative in order to deal with the probabilistic constraint. %Chance constraints have been used extensively in power system operations, including~\cite{ZhangEtAl2011,WuEtAl2014, MartinezEtAl2014, WangYiShenEtAl2017,PanLiAsilomar2017}.  A second challenge is that the actual probability distribution of the uncertainties is almost never known in practice and has to be estimated or approximated using historical samples, adding to the computational challenge~\cite{YuryEtAl2016}.

However, in all of these papers, the authors focus on solving the optimization problem, instead of explicitly characterizing the feasible region due to uncertainty. We note that it is actually an important topic in some other engineering areas to find a good approximation of the feasible region. In integrated circuit design, the problem of design centering~\cite{harwood2017solve,director1977simplicial}, i.e., finding good design parameters inside the feasible region, is crucial in ensuring a good yield in the presence of statistical uncertainties. %There are two potential drawbacks: first, it is hard for the operator to have an idea of how critical or how safe the current operation is since there is no visual characterization of the feasible region; second, if the objective functions or any of the additional constraints change (i.e., different configuration or requirement of the system), the algorithms need to be re-run, even if the changes are small.
There are two potential benefits of having an approximation of the feasible region: first, the grid operator can evaluate how safe or crucial the current operation since there is a visualization of the region; second, the algorithm does not need to be re-run when the objective functions or any of the additional constraints change (i.e., different configuration or requirement of the system).

In this paper, we provide an explicit ellipsoidal approximation of the feasible region bounded by the chance constraints of interest. We choose ellipsoids due to their convex geometry and tractability when incorporated in convex objective functions~\cite{StephenEtAl2004}.
In addition, we provide a sampling algorithm that efficiently trains the proposed model. Compared to random sampling strategy, our algorithm achieves smaller estimation error with comparable number of samples.

Our contributions to ellipsoidal approximation in chance constrained optimization problems are:
\begin{itemize}
	\item Our proposed method is \emph{data driven}: we do not assume any prior knowledge on the probabilistic distribution of the uncertainty in the system.
	\item We provide a novel view on approximating chance constrained optimizations by applying a \emph{machine learning approach}. 
	\item We present an efficient \emph{training algorithm that achieves a small estimation error}.
\end{itemize}

The rest of the paper is organized as follows. In Section \ref{subsec:literature}, we review the literature on ellipsoidal approximation of feasibility sets and introduce our proposed method. In Section \ref{sec:prob}, we present the mathematical formulation of voltage control in power distribution systems. In Section \ref{sec:learning} we find an approximation of the feasible region by a machine learning approach, and in Section \ref{sec:querying} we show how this learning model can be efficiently trained using active sampling techniques. Simulation results based on the IEEE standard bus system are shown in Section \ref{sec:sim}. Conclusions are presented in Section \ref{sec:conclusion}.

\subsection{Prior results}\label{subsec:literature}
The estimation and approximation of feasibility sets using simple geometric structures have been addressed in many applications.     

Some early work focuses on linear approximation of such sets, for example, simplices were used to approximate the feasible region in integrated circuit design~\cite{director1977simplicial}, and the authors of~\cite{SapatnekarEtAl1993} propose to use supporting hyperplanes for a similar purpose.
%However, such methods based on linear approximation can be intractable because we need many hyperplanes to support a smooth convex set, for example, an ellipsoid. \textcolor{red}{cite more.}
More recently, in~\cite{ZhouEtAl2015}, the authors propose to approximate the region of interest by convex piecewise linear machine. In~\cite{JinEtAl2017}, the authors use polytopes to estimate feasible sets of a model predictive control problem.
%Another line of research is on finding smooth representations, i.e., ellipsoids to approximate the feasible region. For example...
A problem of these approaches is that linear approximation may end up with many piecewise segments, leading to exponentially many constraints. %\cite{?}

An alternative to linear approximation is to use more versatile and smooth functions, for example a polynomial function, among which ellipsoids are commonly adopted.
% the most convenient and \textcolor{red}{\cite{?}}
In~\cite{AnstreicherEtAl1999}, the authors propose to approximate convex sets by ellipsoids based on a modification of the volumetric cutting plane method. Authors in~\cite{KiselevEtAl2008} evaluate the optimality of ellipsoidal approximation by moments of the support function. This approach unfortunately requires the solution of a high dimensional integral. In~\cite{DabbeneEtAl2003}, the authors use ellipsoids to approximate polytopes. Ellipsoidal approximation is also widely used in control problems to describe the attraction domain~\cite{PolyakEtAl2009} and to approximated polyhedral sets in power systems~\cite{SaricEtAl2008}.%any work on power??

However, either the metrics are hard to compute in a high dimensional setting, or the close form of the target convex set is required in previous papers on ellipsoidal approximation. Instead, we consider a \emph{data-driven} approach and propose an efficient sampling algorithm that finds the boundary of an unknown feasibility region. %to assist physical system operation, for example, the power system operation under voltage constraints.
By doing so, we do not require any specific knowledge of the geometry on the feasible region, except that it is a bounded set and it is easy to check whether a point is inside the set or not. %In fact, our approach is more reasonable to deal with feasible regions considering real physical systems. For example, the feasible region for voltage in power system considering random renewable energy integration involves multi-dimensional integral, and is thus intractable to obtain a parameterized form~\cite{li2017distribution}. However, given enough historical observations of the system, it is relatively easy to check whether a specific voltage violates the requirement. This check action is referred to as one query from the oracle.

More specifically, we use a support vector machine (SVM) approach to identify the boundary of the feasibility set by an ellipsoid.
%We aim to classify the interior of the set as feasible and the rest as infeasible.
% SVM is ideal for such tasks because the classifier relies on only a small subset of sample points, namely the support vectors, and does not focus on points that are farther away from the decision boundary. %To increase the capacity in depicting the boundary, we use positive semidefinite kernels, which means that the approximated boundary is an ellipsoid.
Since we assume that checking whether a point is feasible or not is relatively inexpensive, we seek to collect favorable sample points and update the learned boundary using those queried points.

Our approach stems from the field of \emph{active learning}~\cite{TongEtAl2001}. Unlike passive learning where samples are collected regardless of the machine learning model itself, in active learning, the samples are selected and queried from the oracle manually to optimally train a specific model. The goal is to query a small number of samples to achieve similar accuracy as passive learning. Active Learning is a rich field with many results obtained over the years, we recommend the interested reader to access ~\cite{BalcanEtAl2009,DasguptaEtAl2006,BalcanEtAl2007} and the references within, to see the label (query) complexity of different sampling methods for either the consistent or the agnostic case. %Different sampling strategies are introduced and discussed in~\cite{}.
%In this paper, we specifically focus on active sampling to train a SVM-like model. In~\cite{TongEtAl2001}, the authors discuss several active learning methods applied to SVM models with different criterion. In~\cite{ZomerEtAl2004}, authors discuss practical issues on how to select a subset of samples points to query in SVM. Besides, authors in~\cite{KremerEtAl2014} discuss more on sample selection in SVM and how to deal with the bias in samples. In this work, we incorporate the idea of active learning in SVM to assist the approximation of a bounded set by aggressively query important sample points.

\section{Problem Set-up}\label{sec:prob}

\subsection{Power flow model for radial networks}
We first present the modeling of components in a radial distribution network in power systems. For interested readers, please refer to \cite{FarivarEtAl2012, GanetAl2012} for more details. We specifically consider a linear approximation of the system, which is presented in~\cite{Wu89}, assuming that the losses are negligible and the voltage at each bus is close to 1 in p.u. \cite{ZhuEtAl2016}. %This enables us to approximate $V_i^2 - V_k^2 $ by $2(V_i - V_k)$ \cite{ZhuEtAl2016}. %As stated in \cite{Low13CDC} and further validated by \cite{ZhuEtAl2016}, the approximation error is small. Those approximations lead to a small error at about 0.25\% if there is a 5\% deviation in voltage magnitude \cite{ZhuEtAl2016}.
We consider a distribution network with $d + 1$ buses ordered as $ 0, 1, \cdots , d$, where bus $0$ is the feeder (reference bus). The linearized relationship between voltage and power injection of the system is given by:
\begin{equation}\label{system}
	\bm{v} = \mathbf{R}\bm{p} + \mathbf{X}\bm{q} + v_0\cdot\bm{1}
\end{equation}
where $\bm{1}$ is an all-one vector, $v_0$ is the nominal voltage value at reference bus $0$, $\bm{v}, \bm{q}, \bm{p} \in \mathcal{R}^{d}$ represent respectively the voltage, reactive power injection and real power injection at each of the buses. The matrices $\mathbf{R}$ and $\mathbf{X}$ are associated with the resistance and reactance of the distribution system \cite{LietAl2014}. %Following the findings in \cite{LietAl2014}, we give the expressions of $R_{ik}$ and $X_{ik}$ in terms of line resistance $r_{ik}$ and reactance $x_{ik}$:
%\begin{equation}
%\label{RXelement}
%R_{ik} = \sum_{(h,l) \in \mathcal{P}_i \cap \mathcal{P}_k} r_{hl}, \ X_{ik} = \sum_{(h,l) \in \mathcal{P}_i \cap \mathcal{P}_k} x_{hl},
%\end{equation}
%where $\mathcal{P}_i$ is the set of lines on the unique path from bus 0 to bus $i$.

\subsection{Voltage control in power distribution systems}

%\subsection{General problem}
%Suppose that there is an unknown compact set that we want to characterize (its boundary more specifically). There is an oracle that can tell us \emph{efficiently} whether a point $x$ is inside the set (with label $y = -1$) or outside the set (with label $y = 1$). Using this oracle, we want to query the least possible points to estimate the boundary.
To facilitate analysis, rewrite $\bm{v}$ as the difference between the bus voltage and the reference voltage $v_0$, then \eqref{system} becomes:
\begin{equation}
	\bm{v} = \mathbf{R}\bm{p} + \mathbf{X}\bm{q}.
\end{equation}

Installment of DER at each of the buses produces uncertainty in power injection, resulting in uncertainty in voltages. The voltage profile is reformulated into the following form:
\begin{equation}\label{noisyV}
	\begin{aligned}
		\bm{v} & = \mathbf{R}(\bm{p} + \Delta \bm{p}) + \mathbf{X}\bm{q} \\
		& = \mathbf{R}\bm{p} + \mathbf{X}\bm{q} + \mathbf{R}\Delta \bm{p} \\
		& = \mathbf{R}\bm{p} + \mathbf{X}\bm{q} + \bm{\epsilon},
	\end{aligned}
\end{equation}
where $\Delta \bm{p}$ is the uncertain power injection due to DER at each bus, and $\bm{\epsilon}$ {represents the uncertainty of the system} that has covariance matrix $\bm{\Sigma}$. {We assume} that $\bm{\epsilon}$ has a continuous density distribution. The covariance matrix $\bm{\Sigma}$ is not a diagonal matrix, even when $\Delta \bm{p}$ has i.i.d. distribution. %For example, uncertainty can be a result of randomness from {renewable generation} at various buses in the distribution network. {This uncertainty causes fluctuations in real power injection and thus bring up voltage fluctuations as well. In all, we use $\bm{\epsilon}$ to capture randomness in voltage from all possible sources.} The randomness is usually highly correlated across buses because of geometrical adjacency. In this case, every bus is dependent of each other due to correlated randomness and $\bm{\epsilon}$ captures the uncertainty of the whole system. %\todo{Much more explanation of this uncertainty. Explain the physical process. Need at least 3 or 4 sentences here.}

The voltages in the system should be maintained within a tight bound usually plus/minus 5\% of nominal. With uncertainties introduced by the operation of the DERs, we model this constraint in a probabilistic fashion. We use the following chance constraint which bounds the probability of the voltages staying in the prescribed bounds:
\begin{equation}\label{max00}
	\Pr \{ \underline{\bm{v}}  \leq \bm{v} \leq \overline{\bm{v}} \} \geq \alpha,
\end{equation}
which is equivalent as:
\begin{equation}\label{max0}
	\Pr \{ \underline{\bm{v}}  \leq \mathbf{R}\bm{p} + \mathbf{X}\bm{q} + \bm{\epsilon} \leq \overline{\bm{v}} \} \geq \alpha,
\end{equation}
where $\underline{\bm{v}} $ and $\overline{\bm{v}}$ are the voltage bounds. The value of $\alpha$ is a parameter that can be chosen to indicate the probability that event $\underline{\bm{v}}  \leq \mathbf{R}\bm{p} + \mathbf{X}\bm{q} + \bm{\epsilon} \leq \overline{\bm{v}}$ occurs. %

In this paper we only consider reactive power regulation and assume that the active load injection $\bm{p}$ is determined exogenously and the controllable variable is the reactive power injection $\bm{q}$. Denote $\Pr \{ \underline{\bm{v}}  \leq \mathbf{R}\bm{p} + \mathbf{X}\bm{q} + \bm{\epsilon} \leq \overline{\bm{v}} \}$ by $g(\bm{q})$, for a given tolerance level $\alpha$, the centralized voltage regulation problem is then captured as the following:
\begin{subequations}\label{prob:main}
	\begin{align}
		&  \min_{\bm{q}} \   C(\bm{q}) \label{eqn:main_obj}\\
		\mbox{s.t.}\ & g(\bm{q}) \geq \alpha, \label{eqn:alpha} %\\
		%& \underline{\bm{q}} \leq \bm{q} \leq \overline{\bm{q}}, \label{eqn:Qlimit}
		%\\ & g(\bm{q}) <  \alpha,
		%\\ & g(\bm{q}) = \alpha,
	\end{align}
\end{subequations}
where the cost function $C(\bm{q})$ can be any convex cost function, for example $\Vert \bm{q} \Vert_2$ \cite{KekatosEtAl2015}. This cost function encourages small amount of reactive power support to maintain the acceptable voltage deviation due to uncertainty. In addition, $\bm{q}$ can be subject to box constraints due to limit of regulatory power.

Computing $g(\bm{q})$ requires computing an integral on the density of $\bm{\epsilon}$ on a multi-dimensional box, which does not generally have a close form. Thus it is hard to parameterize the geometry of the constraint $g(\bm{q}) \geq \alpha$. %Even if $\bm{\epsilon}$ follows the most common ellipsoidal distribution, i.e., multivariate normal distribution, an integral of $\bm{\epsilon}$ over a simple box is still not ellipsoidal in general, see Fig. \ref{contour}.
\begin{comment}

\begin{figure}[!ht]
\centering
\includegraphics[width=0.75\columnwidth]{contour_constraint.eps}
\caption{Contour of $f(x) = \int_{x-1}^{x}p(u)du$, where $p(u)$ is the pdf of $U$, $U \sim \mathcal{N}(0, \Sigma)$ and $\Sigma = [1;-0.063;-0.063;0.05]$.}
\label{contour}
\end{figure}
\end{comment}
{However, its shape is reminiscent of a perfect ellipsoid that motivates us to approximate the region with an ellipsoid. Once the region is approximated, the optimization problem for voltage regulation becomes:
	\begin{subequations}\label{prob:convex}
		\begin{align}
			&  \min_{\bm{q}} \   C(\bm{q}) \\
			\mbox{s.t.}\ & f(\bm{q}) \leq 0 , %\\
			%& \underline{\bm{q}} \leq \bm{q} \leq \overline{\bm{q}},
		\end{align}
	\end{subequations}
	where $ f(\bm{q}) \leq 0$ represents the approximating ellipsoid. If $C(\bm{q})$ is convex, then \eqref{prob:convex} is a convex program and can be solved efficiently \cite{StephenEtAl2004}.
	
	Therefore, the central question in this paper is:  for a given tolerance level $\alpha$, \emph{how can one approximate the feasible region} $\{ \bm{q}, \mbox{s.t.}\  g(\bm{q}) \geq \alpha\}$ \emph{by an ellipsoid quickly and efficiently?} In Section \ref{sec:learning}, we discuss more details on how we answer this question by a SVM machine learning approach, and in Section \ref{sec:querying} we present an efficient training algorithm. 

\section{The learning problem}\label{sec:learning}

 Finding an approximate ellipsoid is non-trivial, since we do not know the distribution of the randomness in the system. Even if the distribution is known, the problem is still hard due to the multivariate integral computation. However, if we can empirically evaluate the feasibility and infeasibility of sufficiently many points in the space, eventually we obtain a reasonable estimate of the feasible region. This evaluation can be done by an \emph{oracle} defined as the following:
 
 \noindent\textbf{Oracle:} An oracle can efficiently check whether $g(\bm{q}) \geq \alpha$ (empirically) given a particular $\bm{q}$. If $g(\bm{q}) \geq \alpha$, the oracle returns a label $-1$, otherwise a label $1$. We say that the point $\bm{q}$ is \textit{queried}, when the oracle checks whether the constraint holds for a specific $\bm{q}$. The query process is illustrated in Fig.~\ref{oracle}.
 \begin{figure}[!ht]
 	\centering
 	\includegraphics[width=0.65\columnwidth]{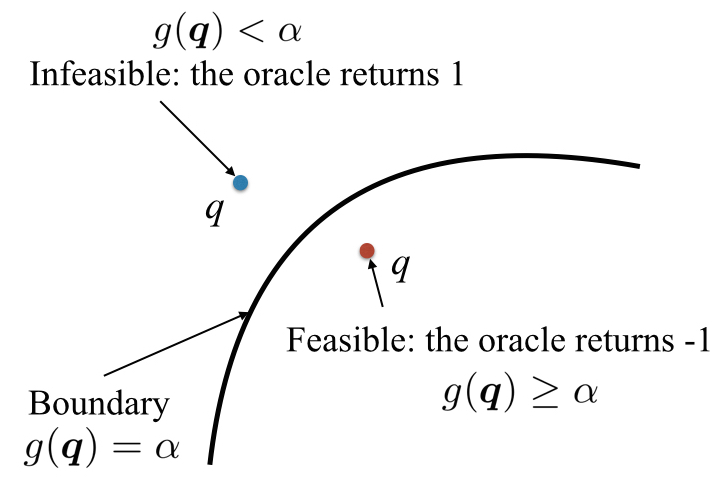}
 	\caption{A point is queried and oracle returns a label indicating whether this point is feasible or not.}
 	\label{oracle}
 \end{figure}
 
 %Thus, an oracle is a black box model that efficiently tells whether a queried $\bm{q}$ lies within the feasible region $g(\bm{q}) \geq \alpha$ or the other way around. This equivalently scopes the boundary of the interested feasible region, i.e., $g(\bm{q}) = \alpha$.
 Given that such an oracle exists, we can efficiently query as many $\bm{q}$'s as possible and hence estimate this boundary accurately. Since a perfect oracle that tests whether exactly $g(\bm{q}) \geq \alpha$ holds requires a multi-dimensional integral, and hence it is computationally challenging, we resort to an approximate test to verify whether $g(\bm{q}) \geq \alpha$. The approximate oracle then tests whether:
 \begin{equation}\label{eqn_sample_approx_g}
 	\begin{aligned}
 		& \hat{g}(\bm{q}, \{\bm{\epsilon}^{(s)}, \forall s \in \mathcal{S}\}) \\
 		\overset{\Delta}{=}  &   \frac{\sum_{s = 1}^{S} \mathds{1}\{ \underline{\bm{v}} \leq \mathbf{R}\bm{p} + \mathbf{X}\bm{q} + \bm{\epsilon}^{(s)} \leq \overline{\bm{v}}\}}{S}\\
 		%= &  \frac{\sum_{s = 1}^{S} \mathds{1}\{ \underline{\bm{v}} \leq \mathbf{X} + \bm{\epsilon}^{(s)} \leq \overline{\bm{v}}\}}{S}\\
 		%\overset{\Delta}{=} &  \widehat{\Pr}(\mathbf{X}, \{\bm{\epsilon}^{(s)}, \forall s \in \mathcal{S}\}),
 	\end{aligned}
 \end{equation}
 is greater than $\alpha$, given that $\bm{\epsilon}^{(s)}$ are i.i.d. samples of $\bm{\epsilon}$. In practice, $\bm{\epsilon}^{(s)}$ can be historical observations of the randomness $\bm{\epsilon}. $Here $ \mathds{1}\{\cdot\}$ is the indicator function and \eqref{eqn_sample_approx_g} is an approximate evaluation of the target chance constraint \eqref{eqn:alpha}.
 
 %$g(\bm{q}) \geq \alpha$,     $g(\bm{q}) < \alpha$,     $g(\bm{q}) = \alpha$,
 
 Let us define $f(\bm{q})\triangleq\bm{q}^{\top}\mathbf{M}\bm{q} + \bm{h}^{\top}\bm{q} + 1$. An ellipsoid can always be represented as $f(\bm{q})\leq 0$, where $\mathbf{M}\succeq 0$. Mathematically, we want to approximate the unknown geometry $\{ \bm{q}\,\vert\,g(\bm{q}) \geq \alpha\}$ by an ellipsoid denoted by $\{ \bm{q}\,\vert\, f(\bm{q})\leq 0, \mathbf{M}\succeq 0\}$ using the oracle.
 Let subscript $n$ denote the index of query.  If we have queried $N$ points, i.e., $\bm{q}_n, \forall n \in \{1,2, \cdots, N\}$, to estimate $\mathbf{M}, \bm{h}$, we solve the following optimization problem:
 \begin{equation}\label{eq:svm}
 	\begin{aligned}
 		\textrm{min}_{\mathbf{M}\succeq 0, \bm{h}} &\ \sum_{n=1}^{N} \nu_n \max\{0, -f(\bm{q}_n, \mathbf{M}, \bm{h})y_n\}
 		%\mbox{s.t.} &\ f(\bm{q}_n,\mathbf{M}, \bm{h}) = \bm{q}_n^{\top}\mathbf{M}\bm{q}_n + \bm{h}^{\top}\bm{q}_n + 1, \\
 		% & \forall n \in \{1,2, \cdots, N\},
 	\end{aligned}
 \end{equation}
 where $y_n$ is the label and $y_n = -1$ if the oracle estimates it is feasible, and $y_n = 1$ otherwise. This formulation is similar to the renowned SVM, except that the loss term is not the hinge loss, i.e., $\max\{0, 1  -f(\bm{q}_n, \mathbf{M}, \bm{h})y_n\}$, but $\max\{0,  -f(\bm{q}_n, \mathbf{M}, \bm{h})y_n\}$. We adopt this loss term because if $f(\bm{q}_n, \mathbf{M}, \bm{h})$ has the same sign as $y_n$, then the loss term is zero.
 
 When the boundary is truly an ellipsoid, the optimization returns an ellipsoid such that no false positive nor false negatives occurs. Otherwise, we tune the weights $\nu_n$'s such that the value is high for false positives in the objective in \eqref{eq:svm}. This minimizes the occurrence of false positives, and the optimization problem approaches the maximum volume ellipsoid \cite{StephenEtAl2004} as number of queried points goes to infinity. In this way, we obtain a feasible solution to \eqref{prob:main} by solving \eqref{eq:svm}. In the next section, we discuss how to efficiently query a small number of $\bm{q}_n $'s to find such an ellipsoid with high accuracy.

\section{Querying procedure}\label{sec:querying}

To illustrate our proposed algorithm, in this section we assume that the  ground truth boundary is an ellipsoid, and we show that only logarithmically many points are needed to guarantee a small error in parameter estimation. More specifically, in Section~\ref{subsec:random_sampling} we present the difference between two sampling procedures for training the model in~\eqref{eq:svm}. In Section~\ref{subsec:binary} we propose an efficient algorithm to sample the necessary points and in Section~\ref{subsec:theorem} we illustrate our main result on the performance of the proposed algorithm. We show in Section~\ref{sec:sim} that the proposed approach also performs well for shapes with non-ellipsoidal boundaries.
%Alberto: what does  empirically mean here?

\subsection{Random sampling and active sampling}\label{subsec:random_sampling}
A naive way to query points in order to estimate $\mathbf{M}, \bm{h}$ from \eqref{eq:svm} is to query a random number of points. However, this approach suffers from two shortcomings: 1). Random sampling queries points in the whole space equally probably. However, points that are further away from the boundary do not contribute to model change as much as the points close to the boundary. 2). If the whole search space is large and the true feasible region is small, the chance that we query a feasible point is relatively low, which does not yield accurate ellipsoidal estimation.

On the other hand, if we query \emph{actively} the sample points that are close enough to the boundary (and diversely located around the boundary), then with the same number of queries that random sampling uses, we should get a better understanding of what the boundary looks like. This is illustrated in Fig. \ref{fig:random_active}.
%Alberto: This process is called weighted Monte Carlo sampling. Please refer to papers or books that deal with this.

\begin{figure}
	\begin{subfigure}[b]{0.5\columnwidth}
		\includegraphics[width=0.8\linewidth]{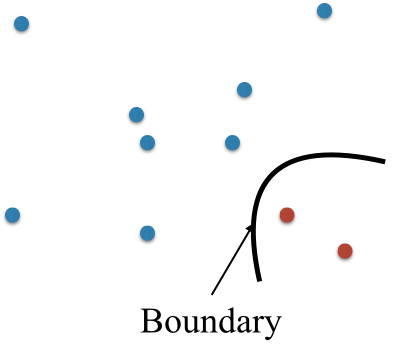}
		\caption{Random sampling.}
	\end{subfigure}%
	\begin{subfigure}[b]{0.5\columnwidth}
		\includegraphics[width=0.8\linewidth]{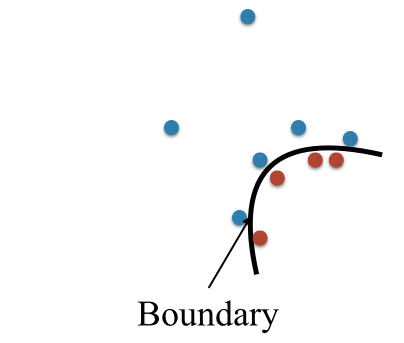}
		\caption{Active sampling.}
	\end{subfigure}
	\caption{Comparison between random sampling and active sampling, where blue dots are queried infeasible points, and red dots are queried feasible points.}
	\label{fig:random_active}
\end{figure}

Now the question becomes: how to query the points such that they are close enough to the boundary while maintaining sufficient diversity?
If the points are not diverse enough, for example the queried points are concentrated inside a small region within the feasible set, we end up with introducing too much bias into the learning model.
%Now there are two main questions. First, how to query the points such that they are close enough to the boundary while maintaining sufficient diversity? %\textcolor{red}{Here we define diversity of the points as the eigenvalue of the data covariance matrix be bounded away from zero.}
%If the points are not diverse enough, for example the queried points are inside a small set within the whole feasible region, then we introduce too much bias into the learning model. Second, if we can ensure diversity, how many queries is enough such that the learned ellipsoid is close enough to the underlying ground truth? %In Section \ref{subsec:theorem}, we first assume that the queried points are sufficiently close to the boundary and diverse enough, i.e., the eigenvalue of the empirical covariance matrix is above zero, and present our main results based on active sampling.
In Section~\ref{subsec:binary}, we present the algorithm to query points such that they are located near the boundary and discuss an easy-to-implement approach to ensure diversity.

%In the next subsection, we proceed on finding a strategy that returns queried points close to the boundary.

\subsection{Binary search}\label{subsec:binary}
We propose a \emph{random binary search} algorithm in Algo.~\ref{algo:active_learning} that queries points with binary cuts along many randomly directions. Overall, the binary cut procedure efficiently locates points close to the boundary with the smallest amount of queries; diversity in queried points is maintained through querying many random directions. In Algo.~\ref{algo:active_learning}, subscript $n$ stands for each instance of queried point. In total, Algo.~\ref{algo:active_learning} requires at most $N_0$ oracle calls. %This enables exploration of the feasible region along random beams and the binary cut procedure ensures that we query points close to the true boundary with minimum number of trys.

\begin{algorithm}[h]
	\caption{Querying algorithm.}
	\label{algo:active_learning}
	\begin{algorithmic}[1]
		\State{\textbf{Input}: A feasible initial point $\bm{q}_0$, ${\tau}>0$, $i = 0$, $n = 1$, $N_0 > 0$.}
		%\While{$y^{(1)} \neq 1$}
		%\State{$\lambda = 2\lambda$.}
		%\State{$\bm{q}^{(1)} = \bm{q}^{(-1)} + \lambda\bm{e}$.}%, $\xi_k = \rho \xi_{k-1}$.
		%\EndWhile
		\While{$n < N_0$}
		\State{$i = i + 1.$}
		\State{Sample a random direction $\bm{e}_i$. Initialize a large enough $\lambda>0$.}
		\State{Let a feasible point $ \bm{q}_i^{(-1)} = \bm{q}_0$. Let an infeasible point $\bm{q}_n  = \bm{q}_i^{(1)} = \bm{q}_i^{(-1)} + \lambda\bm{e}_i$. Query $\bm{q}_n$. Suppose that $\lambda$ is big enough such that the associated label $y_n = 1$.}
		\State{$n =  n + 1$.}
		\While{$\Vert \bm{q}_i^{(-1)} - \bm{q}_i^{(1)} \Vert_2 \geq \tau$}
		\State{Query $\bm{q}_n  = \frac{\bm{q}_i^{(-1)}  + \bm{q}_i^{(1)}}{2}$. Let associated label be $y_n$.}
		\State{ $n =  n + 1.$}
		\If{$y_n = 1$}
		\State $\bm{q}_i^{(1)} = \bm{q}_n$.
		\Else
		\State $\bm{q}_i^{(-1)} = \bm{q}_n$.
		\EndIf
		\EndWhile
		\EndWhile
	\end{algorithmic}
\end{algorithm}

We first explain the inner "while" loop for each fixed direction $\bm{e}_i$ from line 3 to line 13 in Algo. \ref{algo:active_learning}.
The inner loop guarantees that when it terminates, we have at least one feasible point $\bm{q}_i^{(-1)}$ and one infeasible point $\bm{q}_i^{(1)}$ along direction $\bm{e}_i$. Further, their distance from the true boundary is no larger than $\tau$ . If we can gather many pairs of such $\bm{q}_i^{(-1)}$'s and $\bm{q}_i^{(1)}$'s which are diverse enough, we should be able to obtain a good estimation of the true ellipsoid, illustrated by Fig. \ref{fig:random_active}. In Theorem \ref{main_theorem}, we state that to achieve a small error in parameter estimation, it suffices to query logarithmically such points.

%As in the assumption in Theorem \ref{main_theorem}, we also need diversity in queried points. This suggests that sampling points along similar directions is not enough to explore the whole space.
To ensure maximum diversity in queried points, we generate sufficiently many \emph{random} directions $\bm{e}_i$'s and run Algo. \ref{algo:active_learning}. %such that the number of queried points exceeds that is shown in Theorem \ref{main_theorem}.
This is described by the outer "while" loop, which loops over many different random directions $\bm{e}_i$'s. Once the points are collected using Algo. \ref{algo:active_learning}, we can train the model in \eqref{eq:svm}. Considerations on estimation errors are provided in Section \ref{subsec:theorem}.

\subsection{Main result}\label{subsec:theorem}

%, and provide more discussions in Section \ref{subsec:random_sampling} and Section \ref{subsec:binary}.
Before presenting the main result in this paper, we first introduce some assumptions and notations used in this section.
\noindent\textbf{Assumption 1}: The ground truth boundary $\{ \bm{q}, {\mbox{s.t.}}\  g(\bm{q}) = \alpha\}$ is an ellipsoid described by $\{ \bm{q}, {\mbox{s.t.}}\ \bm{q}^{\top}\mathbf{M}^*\bm{q} + (\bm{h}^*)^{\top}\bm{q} + 1 = 0\}$.

\noindent\textbf{Notation 1}: Using subscript $n$ for each query, we query $N$ points, i.e., $\bm{q}_n, \forall n \in \{1,2, \cdots, N\}$.

\noindent\textbf{Notation 2}: Let us adopt notation $\text{vec}(A) \overset{\Delta}{=} [A_{11}, A_{12}, ..., A_{mn}], \text{when} \ A \in \mathcal{R}^{m \times n}$. Let $\bm{z}_n = [\bm{q}_n; \text{vec}(\bm{q}_n\bm{q}_n^{\top})]$, and $\mathbf{Z} = [\bm{z}^{\top}_1; \bm{z}^{\top}_2; \cdots, \bm{z}^{\top}_N]$.

% 	\begin{theorem}\label{main_theorem}
%	Let $c_0, c_1, c_2$ be some constants and $\bm{q} \in \mathcal{R}^{d}$. Let us follow Assumption 1 and Notation 1. Using subscript $n$ for each query, we query $N$ points, i.e., $\bm{q}_n, \forall n \in \{1,2, \cdots, N\}$, to train the model in \eqref{eq:svm}. We further let $\bm{z}_n = [\bm{q}_n; \text{vec}(\bm{q}_n\bm{q}_n^{\top})]$, and $\mathbf{Z} = [\bm{z}^{\top}_1; \bm{z}^{\top}_2; \cdots, \bm{z}^{\top}_N]$. Let $c_0 = 2\max_n \Vert \bm{q}_n\Vert_{\infty}$, and $\Vert [{\bm{h}}; \text{vec}(\mathbf{M})] \Vert_2 < c_1$.  Let $\hat{\mathbf{M}}, \hat{\bm{\bm{h}}}$ be obtained by solving \eqref{eq:svm} using the queried samples according to sampling algorithm, i.e., Algo. \ref{algo:active_learning} with sufficient diversity, i.e., $\lambda_{\text{min}}(\mathbf{Z}^{T}\mathbf{Z})/ N = c_2 $ is bounded away from zero. Then %under the sampling strategy shown in Algo. \ref{algo:active_learning},
%	with $N > N_0 = \frac{c_1^2c_0^2}{c_2^2}\log(\frac{1}{\tau})O(d^2)$, the estimation in parameter of the SVM shown in \eqref{eq:svm} is less than $\tau$ in 2-norm, i.e., $\Vert [\hat{\bm{h}}; \text{vec}(\hat{\mathbf{M}})]  - [{\bm{h}^*}; \text{vec}(\mathbf{M}^*)] \Vert_2 \leq \tau$.
%	\end{theorem}

We now present our main result on sampling complexity in Theorem \ref{main_theorem}.

\begin{theorem}\label{main_theorem}
	Let $\bm{q} \in \mathcal{R}^{d}$ and Assumption 1, Notation 1 and 2 hold. Let $\hat{\mathbf{M}}, \hat{\bm{\bm{h}}}$ be obtained by solving \eqref{eq:svm} using the queried samples according to Algo. \ref{algo:active_learning} with $\lambda_{\text{min}}(\mathbf{Z}^{T}\mathbf{Z})/ N$ bounded away from zero. Then %under the sampling strategy shown in Algo. \ref{algo:active_learning},
	with $N > N_0 = O(d^2\log(\frac{1}{\tau}))$, the estimation error in parameters of the model in \eqref{eq:svm} is less than $\tau$ in 2-norm, i.e., $\Vert [\hat{\bm{h}}; \text{vec}(\hat{\mathbf{M}})]  - [{\bm{h}^*}; \text{vec}(\mathbf{M}^*)] \Vert_2 \leq \tau$, when $\tau$ is sufficiently small.
\end{theorem}

Proof of Theorem \ref{main_theorem} is in Appendix \ref{proof:main_theorem} for interested readers. From Theorem \ref{main_theorem}, we first observe that under the assumption that the true boundary is an ellipsoid, Algo. \ref{algo:active_learning} achieves a logarithmic bound in estimation error on the number of queries, as compared to the standard result that random sampling achieves a linear bound \cite{EhrenfeuchtEtAl1989,Hanneke2016}. This improvement in sampling complexity is due to the fact that Algo. \ref{algo:active_learning} \emph{actively} selects the best points to query instead of blindly querying random points. %A detailed comparison between active sampling and random sampling is discussed in Section \ref{subsec:random_sampling}.
Second, to achieve the bound, the queried points have to be diverse, i.e., the smallest eigenvalue of the data covariance matrix is bounded away from zero. This enables us to fully explore the unknown feasible region, and reduces bias in the learning model.% In Section \ref{subsec:binary} we present Algo. \ref{algo:active_learning} that maintains diversity while actively querying the best points.

\section{Simulation}\label{sec:sim}
%\subsection{Synthetic examples}

%\subsection{IEEE 13 bus system}
In this section, we validate our approach using the IEEE standard test feeder. Here we use IEEE 13 bus feeder \cite{IEEE123busref}. The test feeder is shown in Fig. \ref{13bus}, where we assume bus 1 is the reference bus. More results based on synthetic data are left to Appendix \ref{subsec:synthetic} for interested readers.
% The other buses are of interest in the computation. In this test feeder, there is a transformer between 5 and 6 and the impedance is calculated on the higher voltage side. Bus 10 is a switch and we assume that it is closed. Ignoring the impedance between bus 9 and bus 10, these two buses are considered as one single bus in the system.
\begin{figure}[!h]
	\centering
	\includegraphics[width=0.75\columnwidth]{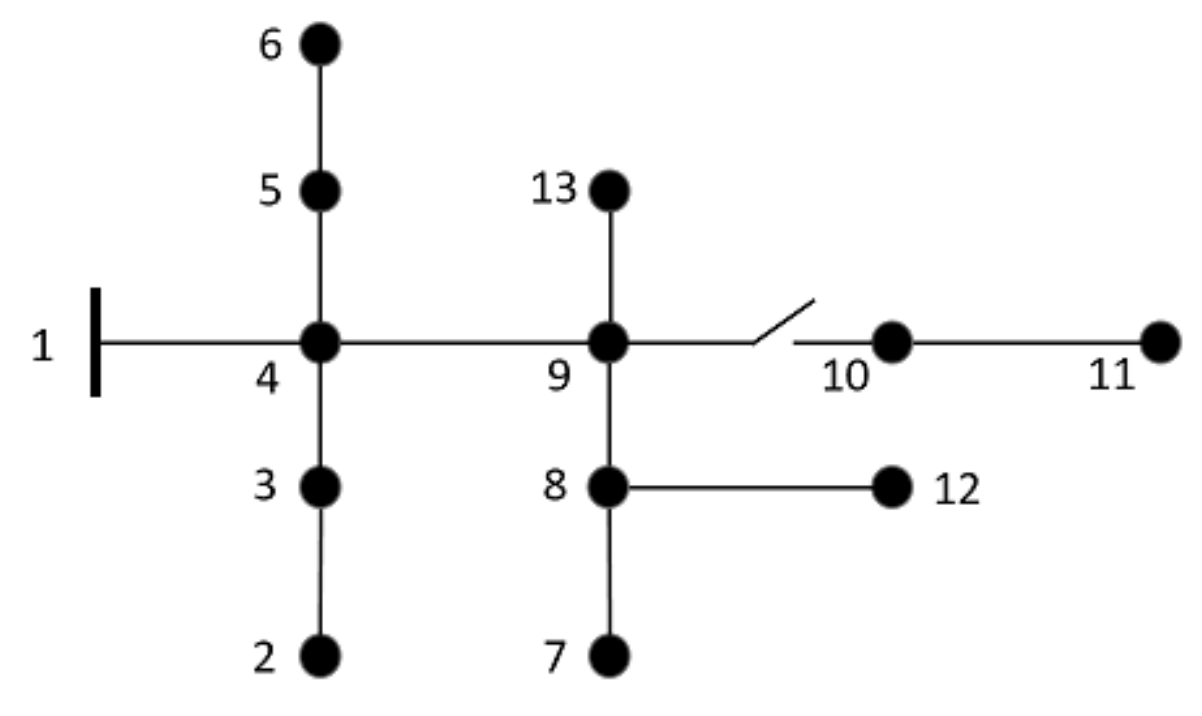}
	\caption{Schematic diagram of IEEE 13 bus test feeder.}
	\label{13bus}
\end{figure}

In this distribution test feeder, we assume that there is no distributed generation, so the active power injection on each bus is negative. The line impedance is retrieved from \cite{IEEE123busref}. In addition, we restrict the available reactive power regulation at the buses to be no more than 0.1 (p.u.). We assume that $\overline{\bm{v}} = 0.5$ and $\underline{\bm{v}} = -0.5$. The randomness $\bm{\epsilon}$ presented in this system is assumed to be a multivariate normal distribution with zero mean and high degrees of correlation, i.e., bus 2, 3, 9, 10, 12, 13 have non-trivial correlations. The distribution itself is unknown by the algorithm; however, historical observations of $\bm{\epsilon}$, i.e., $\bm{\epsilon}^{(s)}$'s are available.

%\bjin{Remove the $1/2$ coefficient in the cost function?}
Let us use the squared 2-norm as the cost of the optimization problem, i.e., $C(\bm{q}) = \frac{1}{2}\Vert \bm{q} \Vert_2^2$. The benchmark to solve the chance constrained problem in \eqref{prob:main} is the scenario approximation in \cite{LuedtkeEtAl2008}:
\begin{subequations}\label{prob:main_SAA}
	\begin{align}
		&  \min_{\bm{q}} \  \Vert \bm{q} \Vert_2^2 \label{eqn:main_obj_SAA}\\
		s.t.\ & \sum_{s \in \mathcal{S}} \mathds{1} \{\underline{\bm{v}} \leq \mathbf{R}\bm{p} + \mathbf{X}\bm{q} + \bm{\epsilon}^{(s)} \leq \overline{\bm{v}}\} \geq \alpha S \label{eqn:alpha_SAA}. %\\
		%& \underline{\bm{q}} \leq \bm{q} \leq \overline{\bm{q}}.
	\end{align}
\end{subequations}

The indicator functions in this optimization problem can be replaced by ancillary binary variables and \eqref{prob:main_SAA} can be solved by mixed integer programming (MIP). Details of the MIP formulation are presented in \cite{LiEtAl2017}. The results comparing MIP and  the proposed ellipsoidal approximation are summarized in Table \ref{table:13bus}.

\begin{table}[!h]
	\renewcommand{\arraystretch}{1.3}
	\caption{Comparison between MIP in \eqref{prob:main_SAA} and Ellipsoidal approximation in in \eqref{prob:convex} for IEEE 13 bus.}
	\label{table:13bus}
	\centering
	\begin{tabular}{|p{3cm}|p{1.5cm}|p{3cm}|}
		\hline
		\bfseries   &\bfseries MIP  & \bfseries Ellipsoidal approximation \\
		\hline
		Running time (seconds)  & {$519$}  & $19$\\
		\hline
		Empirical risk level & $10.00\%$  & $9.20\%$\\
		\hline
		$\Vert\bm{q}\Vert_2$ (p.u.)  & $0.1321$ &  $0.1352$ \\
		\hline
	\end{tabular}
\end{table}

From Table \ref{table:13bus}, we observe that using ellipsoidal approximation significantly reduces computational time. It also yields a comparable result to that from MIP, with a slightly conservative risk level. In addition, MIP is inefficient in finding a sub optimal solution to \eqref{prob:main_SAA}, yielding a relative duality gap of 58.4 \% after 42 seconds, whereas our algorithm finds a sub optimal solution within just 19 seconds. This makes our algorithm more adaptable for real-time operation.
%Alberto: What does sub-optimal solution mean here? 

\section{Conclusion}\label{sec:conclusion}

We proposed a stochastic programming framework to solve voltage regulation problems. In order to approximate the feasible region of the chance constraint, we formulate the approximation procedure using a machine learning framework. We present an efficient active sampling algorithm that only queries points close to the boundary of the chance constraint. We find that this procedure queries logarithmically many points under the assumption that the true region is an ellipsoid. We extend this result to non-ellipsoidal regions using IEEE standard test feeders. Simulation results show that our proposed algorithm has significant better performance than standard approaches. In the future, we aim to extend the result to larger power distribution systems.

\bibliographystyle{IEEEtran}	% (uses file "plain.bst")
\bibliography{IEEEabrv,paperref}

\appendix

\subsection{Proof for Theorem \ref{main_theorem}}\label{proof:main_theorem}

\begin{proof}
	Note that according to the assumption, the boundary is an ellipsoid, i.e., $\bm{q}^{\top}\mathbf{M}\bm{q} + \bm{h}^{\top}\bm{q} + c = 0$ is the boundary, where $c = 1$ if the origin is inside the ellipsoid and $c = -1$ if the origin is outside the ellipsoid. WLOG, let us assume here that $c = 1$.
	
	Let us suppose that we have collected $\tilde{N}$ pairs of data points $\bm{q}_i^{(-1)}, \bm{q}_i^{(1)}$ (feasible and infeasible) that are within $\tau$ distance of each other. To obtain those points, we need in total $N \geq \tilde{N} O(\log(\frac{1}{\tau}))$ queries due to multiple binary cuts \cite{DasguptaEtAl2006}. In the following we obtain a bound on $\tilde{N}$.
	
	Let ($\mathbf{M}_1, \bm{h}_1, 1$) and ($\mathbf{M}_2, \bm{h}_2, 1$) be two ellipsoids that perfectly separate them. Then the boundary of those two hypothesis should satisfy the following: if
	\begin{equation}\label{eq:M1}
		(\bm{q}_i^{(0)})^{\top}\mathbf{M}_1\bm{q}_i^{(0)}+\bm{h}_1^{\top}\bm{q}_i^{(0)}+1 = 0
	\end{equation}
	where $\bm{q}_i^{(0)}$ is on the line segment between $\bm{q}_i^{(-1)}$ and $\bm{q}_i^{(1)}$, and \begin{equation}\label{eq:M2}
		(\bm{q}_i^{(0)}+\Delta \bm{q}_i)^{\top}\mathbf{M}_2(\bm{q}_i^{(0)}+\Delta \bm{q}_i)+\bm{h}_2^{\top}(\bm{q}_i^{(0)}+\Delta \bm{q}_i)+1 = 0
	\end{equation}
	where $\Delta \bm{q}_i$ is along the direction of line segment between $\bm{q}_i^{(-1)}$ and $\bm{q}_i^{(1)}$ and $\Vert \Delta \bm{q}_i \Vert_2 \leq \tau$. An illustration is shown in Fig. \ref{proof}.

	\begin{figure}[!h]
		\centering
		\includegraphics[width=0.65\columnwidth]{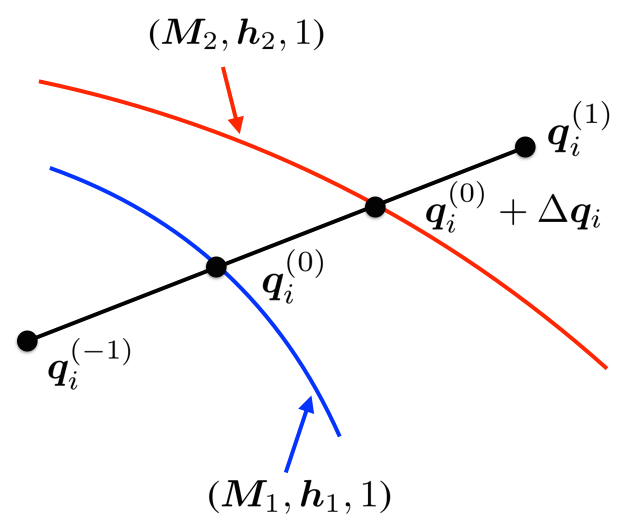}
		\caption{An illustration of \eqref{eq:M1} and \eqref{eq:M2}.}
		\label{proof}
	\end{figure}

	Let $\bm{w}_1$ and $\bm{w}_2$ compactly represent the ellipsoid and $\bm{z}$ denote $ [\bm{q}; \text{vec}(\bm{q}\bm{q}^{\top})]$. Then the decision boundary is compactly represent by $\bm{w}^{\top}\bm{z} + 1$, where $\bm{w} = [\bm{h}; \text{vec}(\mathbf{M})]$. And the above condition can be transformed as:
	\begin{equation}
		\bm{w}_1^{\top}\bm{z}_i^{(0)}+1 = 0,
	\end{equation}
	and 
	\begin{equation}
		\bm{w}_2^{\top}(\bm{z}_i^{(0)}+\Delta \bm{z}_i)+1 = 0,
	\end{equation}
	where $\Delta \bm{z}_i = [ \Delta\bm{q}_i; \text{vec}(2\Delta\bm{q}_i(\bm{q}_i^{(0)})^{\top} +\Delta\bm{q}_i\Delta \bm{q}_i^{\top})]$.
	
	Since $\Vert \Delta \bm{q}_i \Vert_2 \leq \tau$, we know that $\Vert \Delta \bm{z}_i \Vert_{\infty} \leq \max\{\tau, c_0\tau +  \tau^2\}$, where  $c_0 = 2 \max_i \Vert \bm{q}_i^{(0)}\Vert_{\infty}$. Assume that $\tau\ll c_0$, we simply have that $\Vert \Delta \bm{z}_i \Vert_{\infty} \leq c_0\tau$. %$c_0 = \max_i \{\Vert \bm{q}_i\Vert_{\infty}, (\Vert \bm{q}_i\Vert_{\infty})^2\}$.
	
	Subtracting one from another:
	\begin{equation}
		(\bm{w}_1-\bm{w}_2)^{\top}\bm{z}_i^{(0)} - \bm{w}_2^{\top}\Delta \bm{z}_i = 0.
	\end{equation}
	
	Let $\mathbf{Z} = [\bm{z}^{\top}_1;\bm{z}^{\top}_2;...;\bm{z}^{\top}_{\tilde{N}}]$, $\Delta \mathbf{Z} = [\Delta \bm{z}^{\top}_1; \Delta \bm{z}^{\top}_2;...;\Delta \bm{z}^{\top}_{\tilde{N}}] $ and $\bm{w}_1 -\bm{w}_2 = \Delta \bm{w}$. Therefore we have a linear system such that:
	\begin{equation}
		\mathbf{Z} (\Delta \bm{w}) = (\Delta \mathbf{Z})\bm{w}_2. 
	\end{equation}
	
	Let $(\mathbf{Z}^{\dagger})\mathbf{Z} = I$, then 
	\begin{equation}
		\Delta \bm{w} = \mathbf{Z}^{\dagger}(\Delta \mathbf{Z})\bm{w}_2 .
	\end{equation}

	Since the intercept is constrained to be $1$, we know that $\Vert \bm{w}_1 \Vert_2,  \Vert \bm{w}_2 \Vert_2 \leq c_1$ where $c_1$ is a constant, otherwise the ellipsoid is not well defined.

	{Let us bound the above as the following: %Assumptions: %$\Vert \bm{w}_2 \Vert_2 = 1$, $Z^{\dagger}Z = I$, $\Vert \Delta \bm{z}_i\Vert_{\infty} \leq \tau .$%Each row of $\Delta Z$ The linear system is equivalent to write as:
		\begin{equation}
			\begin{aligned}
				\Vert \Delta \bm{w}\Vert_2 &= \Vert \mathbf{Z}^{\dagger}\mathbf{Z}(\mathbf{Z}^{\top}\mathbf{Z})^{-1}\mathbf{Z}^{\top}\Delta \mathbf{Z} \bm{w}_2 \Vert_2\\
				& \leq \Vert \mathbf{Z}^{\dagger}\Vert_2 \Vert \mathbf{Z}(\mathbf{Z}^{\top}\mathbf{Z})^{-1}\mathbf{Z}^{\top}\Delta \mathbf{Z} \bm{w}_2 \Vert_2\\
				& \overset{(a)}{\leq} \sigma_{min}(\mathbf{Z})^{-1}\Vert \mathbf{Z}(\mathbf{Z}^{\top}\mathbf{Z})^{-1}\mathbf{Z}^{\top}\Delta \mathbf{Z} \bm{w}_2 \Vert_2\\
				& \overset{(b)}{\leq} \sigma_{min}(\mathbf{Z})^{-1}\sqrt{D} \tau c_1c_0\\
				& \leq \sqrt{\lambda_{min}(\mathbf{Z}^{\top}\mathbf{Z})}^{-1} \sqrt{D} \tau c_1c_0\\
				& \leq \sqrt{\lambda_{min}(\mathbf{Z}^{\top}\mathbf{Z}/\tilde{N})}^{-1} \sqrt{D/\tilde{N}} \tau c_1c_0\\
			\end{aligned}
		\end{equation}
		where $(a)$ is based on the fact that $\Vert \mathbf{Z}^{\dagger} \Vert_2 = \sigma_{max}(\mathbf{Z}^{\dagger}) = \sigma_{min}(\mathbf{Z})^{-1}$. $(b)$ is based on the fact that $\Vert \Delta \mathbf{Z} \bm{w}_2 \Vert_{\infty} \leq c_1c_0 \tau$ because $\Vert \Delta \bm{z}_i \Vert_{\infty} \leq c_0 \tau$ and $\Vert \bm{w}_2 \Vert_2 \leq c_1$, and that $\mathbf{Z}(\mathbf{Z}^{\top}\mathbf{Z})^{-1}\mathbf{Z}^{\top}$ is a projection matrix with rank $D$ according to Lemma \ref{lemma1} in Appendix \ref{appendix:lemma}.
		
		Given that $c_2 = \lambda_{min}(\mathbf{Z}^{\top}\mathbf{Z}/\tilde{N})$ is bounded away from zero because the sampled points are diverse, we know that in order to have $\Vert \Delta w\Vert_2 \leq \tau$, we need $\tilde{N} \geq \frac{c_1^2c_0^2}{c_2^2} D$. Using the fact that $D = d^2 + d + 1 = O(d^2)$, this indicates that $N \geq \frac{c_1^2c_0^2}{c_2^2} O(d^2) O(\log(\frac{1}{\tau}))  = O(d^2\log
		(\frac{1}{\tau}))$.		
		
		Last, using the fact that $(\hat{\mathbf{M}}, \hat{\bm{h}},1)$ obtained by solving \eqref{eq:svm} perfectly separates the feasible and infeasible points when $(\mathbf{M}^*, \bm{h}^*, 1)$ is truly an ellipsoid, we let the former be represented by ($\mathbf{M}_1, \bm{h}_1, 1$) and the latter be ($\mathbf{M}_2, \bm{h}_2, 1$), this concludes the final proof.
	}
\end{proof}

\subsection{Lemma \ref{lemma1} and its proof}\label{appendix:lemma}
\begin{lemma}\label{lemma1}
	If $P_{\mathbf{Z}} = \mathbf{Z}(\mathbf{Z}^{\top}\mathbf{Z})^{-1}\mathbf{Z}^{\top}$ is a projection matrix where $\mathbf{Z} \in \mathcal{R}^{N\times D}$, where $N > D$, then if a vector $\bm{x} \in \mathcal{R}^{N}$ such that $\Vert \bm{x} \Vert_{\infty} \leq c$, we have that $\Vert P_{\mathbf{Z}} \bm{x} \Vert_2 \leq c\sqrt{D}$.
\end{lemma}
\begin{proof}
	Note that because $P_{\mathbf{Z}} P_{\mathbf{Z}}  = P_{\mathbf{Z}} $, it suffices to show that $\bm{x}^{\top} P_{\mathbf{Z}} \bm{x} \leq c^2 {D}$.
	We write:
	\begin{equation}
		\begin{aligned}
			\bm{x}^{\top} P_{\mathbf{Z}} \bm{x}  &  = \bm{x}^{\top} \mathbf{U}^{\top}\bm{\Sigma U} \bm{x} \\
			&  = (\mathbf{U}\bm{x})^{\top}\bm{\Sigma} \mathbf{U} \bm{x} \\ 
		\end{aligned}
	\end{equation}
	where $\mathbf{U}$ is a orthonormal matrix, and $\bm{\Sigma}$ is a diagonal matrix with first $D$ diagonal elements equals to one and the rest zero.
	
	Since $\Vert \bm{x} \Vert_{\infty} \leq c$ and $\mathbf{U}$ is a orthonormal matrix, we have $\Vert \mathbf{U} \bm{x} \Vert_{\infty} \leq c$. Along with the fact that the diagonal matrix $\bm{\Sigma}$ only has $D$ non zero elements, this indicates that $(\bm{Ux})^{\top}\bm{\Sigma U x}\leq c^2 D$, which concludes the proof.
	
\end{proof}

\subsection{Synthetic examples}\label{subsec:synthetic}
We generate 2-dimensional toy examples to compare the proposed active learning procedure with a standard random sampling procedure. We fix the number of queries to be the same in both procedures. To test the performance, we use the following two geometries (an ellipse and a square) as the ground-truth convex set:
\begin{subequations}
	\begin{align}
		\label{ellipse}
		& \{ \bm{x}\,\vert\,\Vert \mathbf{A}\bm{x} - \bm{b} \Vert_2 \leq 1 \}, \\
		\label{square}
		& \{ \bm{x}\,\vert\,-1 \leq \bm{x} \leq 3\}.
	\end{align}
\end{subequations}

\begin{figure}[H]
	%\centering
	%\raggedleft
	\begin{subfigure}[b]{0.5\textwidth}
		\includegraphics[width=1\linewidth]{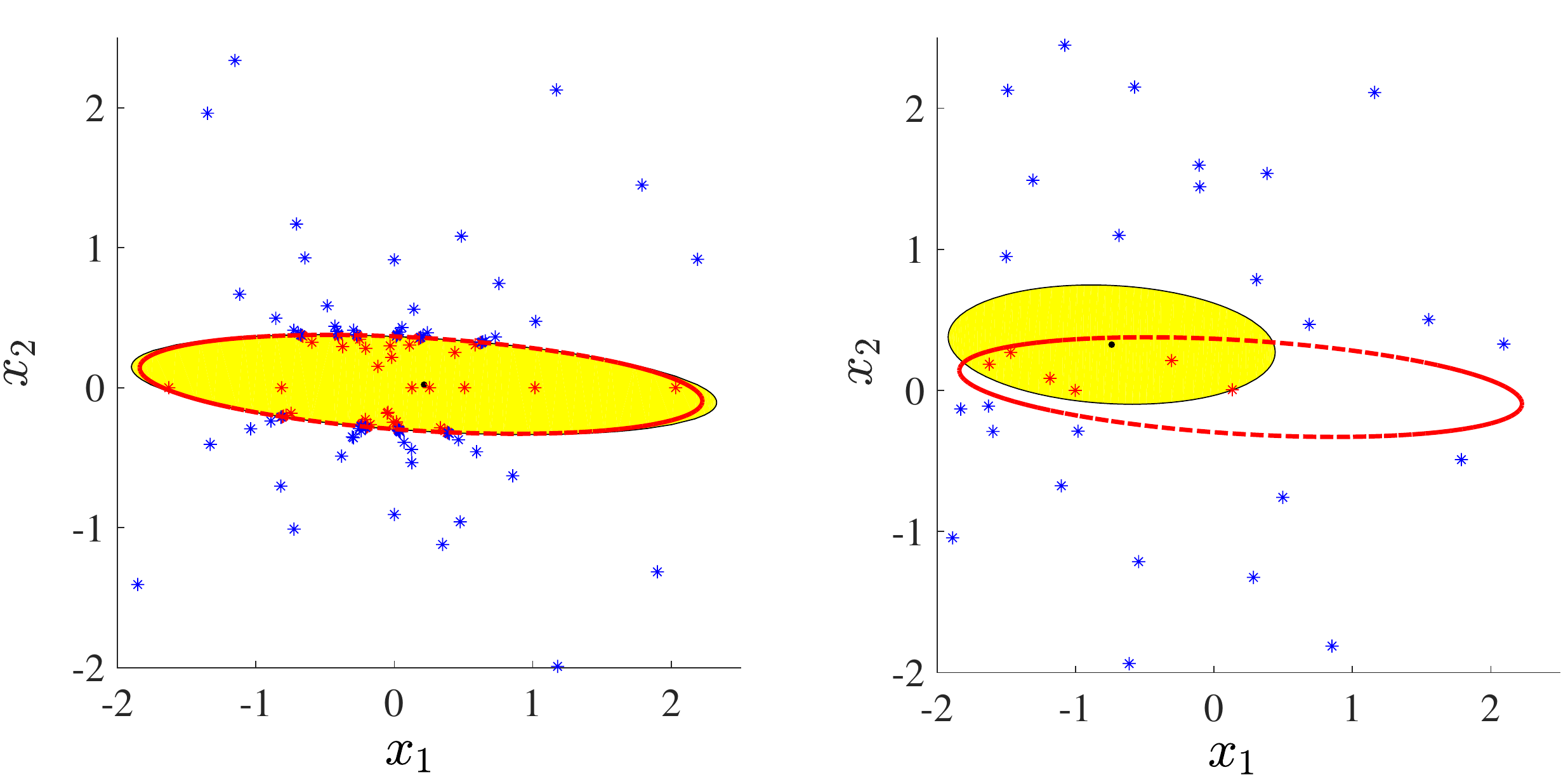}
		\caption{When ground truth is an ellipse according to \eqref{ellipse}. }
		\label{fig:ellipse} 
	\end{subfigure}
	
	\begin{subfigure}[b]{0.5\textwidth}
		\includegraphics[width=1\linewidth]{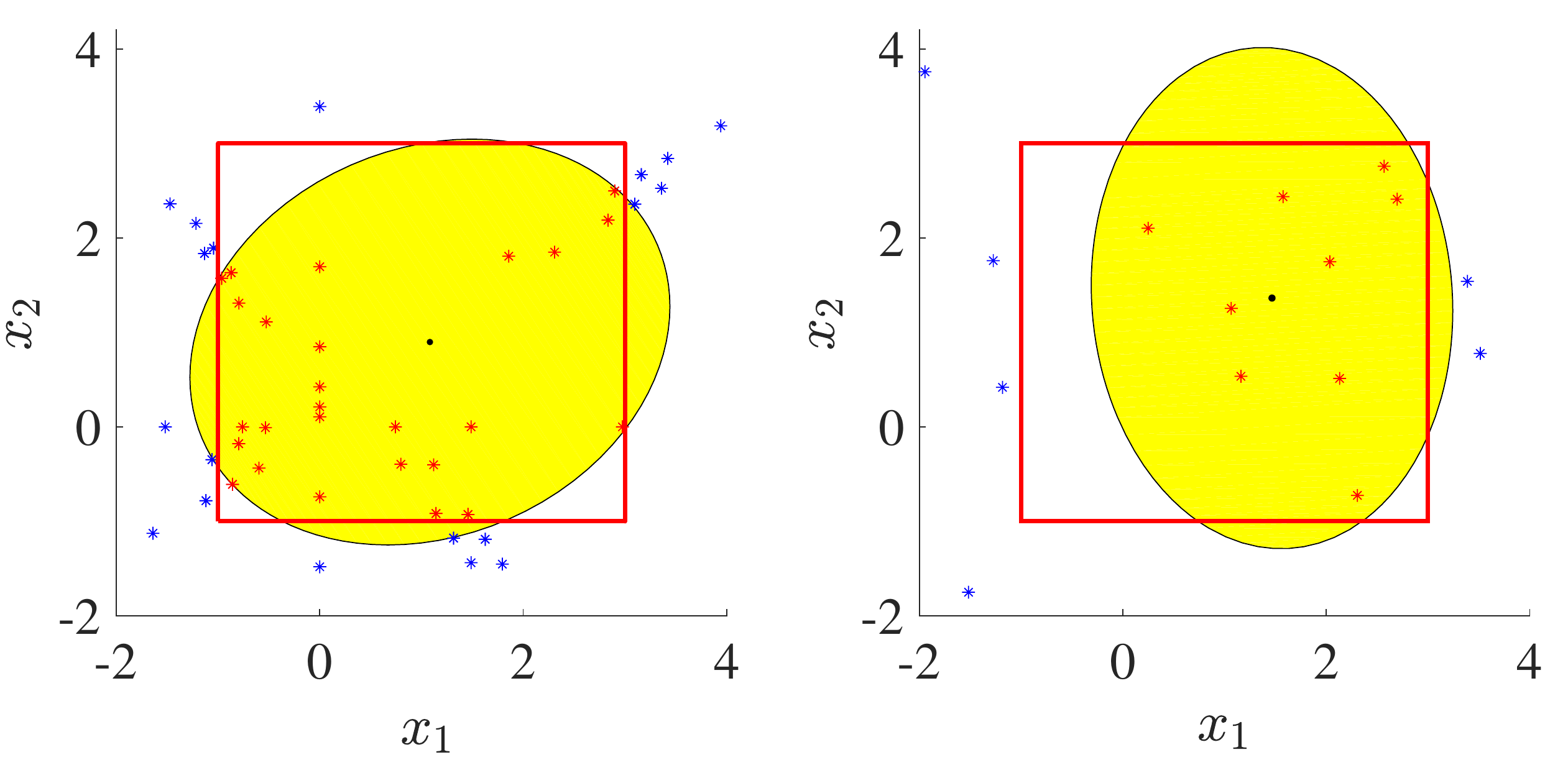}
		\caption{When ground truth is a square according to \eqref{square}.}
		\label{fig:square}
	\end{subfigure}
	
	\caption{Comparison between active sampling and random sampling. The red dots are queried feasible points and the blue dots are queried infeasible points. The ground truth boundary is represented by red solid lines. The fitted ellipse is represented by yellow region. Left: active sampling, right: random sampling.}
\end{figure}

As can be seen from Fig. \ref{fig:ellipse}, active sampling yields a more accurate ellipsoidal approximation as opposed to random sampling. Even with the case where the underlying compact set is not an ellipsoid, active sampling still achieves a much better ellipsoidal approximation than that of random sampling, as shown in Fig. \ref{fig:square}.

\end{document}